\documentclass[a4paper]{llncs}

\usepackage{color}
\usepackage{hyperref}
\usepackage{subfig}
\usepackage{subfloat}
\usepackage[pdftex]{graphicx}

\DeclareGraphicsExtensions{.pdf,.jpeg,.png}

\newcommand{\ww}{\bm{w}}
\newcommand{\uu}{\bm{u}}

\newcommand{\bb}{\bm{b}}
\newcommand{\xx}{\bm{x}}
\newcommand{\yy}{\bm{y}}
\newcommand{\cc}{\bm{c}}

\usepackage{amsfonts,amssymb}
\usepackage{amsmath}
\usepackage{bm}

\usepackage{algorithm}
\usepackage{algpseudocode}

\newtheorem{prop}{Proposition}

\begin{document}

\title{A Generalization of Montucla's Rectangle-to-Rectangle Dissection to Higher Dimensions} 
\author{Antonio Campello\inst{1}\thanks{The work of A.C. was supported by FAPESP grant 2013/25219-5. } \and Vinay A. Vaishampayan\inst{2}} 
\institute{Institute of Mathematics, Statistics and Computer Science \\ University of Campinas \\
\email{campello@ime.unicamp.br} \\ \mbox{} \and
Department of Engineering Science and Physics \\
City University of New York, College of Staten Island \\ \email{Vinay.Vaishampayan@csi.cuny.edu}}
\maketitle

\begin{abstract}
Dissections  of polytopes are a well-studied subject by geometers as well as recreational mathematicians. A recent application in coding theory arises from the problem of parameterizing binary vectors of constant Hamming weight~\cite{tian2009coding}, \cite{sloane2009generalizations}, which is shown to be equivalent to  the problem of dissecting a tetrahedron to a brick.  An application of dissections to a problem related to the construction of analog codes arises in~\cite{campello2013projections}.

Here we consider the rectangle-to-rectangle dissection due to Montucla~\cite{frederickson2003dissections}. Montucla's dissection is first reinterpreted in terms of the Two Tile Theorem~\cite{sloane2009generalizations}. Based on this, a cube-to-brick dissection is developed in $\mathbb{R}^n$.  We present a linear time algorithm (in  $n$) that computes the dissection, i.e. determines a point in the cube given a point in a specific realization of the brick. An application of this algorithm to  a previously reported analog coding scheme~\cite{campello2013projections} is also discussed.
\\[1\baselineskip]
\end{abstract}
\section{Introduction}
In the 18$^{\mbox{th}}$ century, the French mathematician Jean-{\'E}tienne Montucla discovered how to decompose a square into pieces and reassemble them into a rectangle. This process, called a \textit{square-to-rectangle} dissection, follows from a method described in \cite[p.222]{frederickson2003dissections}: 

``Draw (the square) $Q$ parallel to the coordinate axes and extend its basis to the right, placing lines perpendicular to it at intervals equal to the length of the base of $Q$. Draw (the rectangle) $Q'$ so that its upper left corner coincides with the upper left corner of $Q$ and its upper right corner falls on the extension of the base of $Q$. Place a cut in $Q'$ wherever a line crosses it.''

\begin{figure}[!htb]
\centering
\includegraphics[scale=0.6]{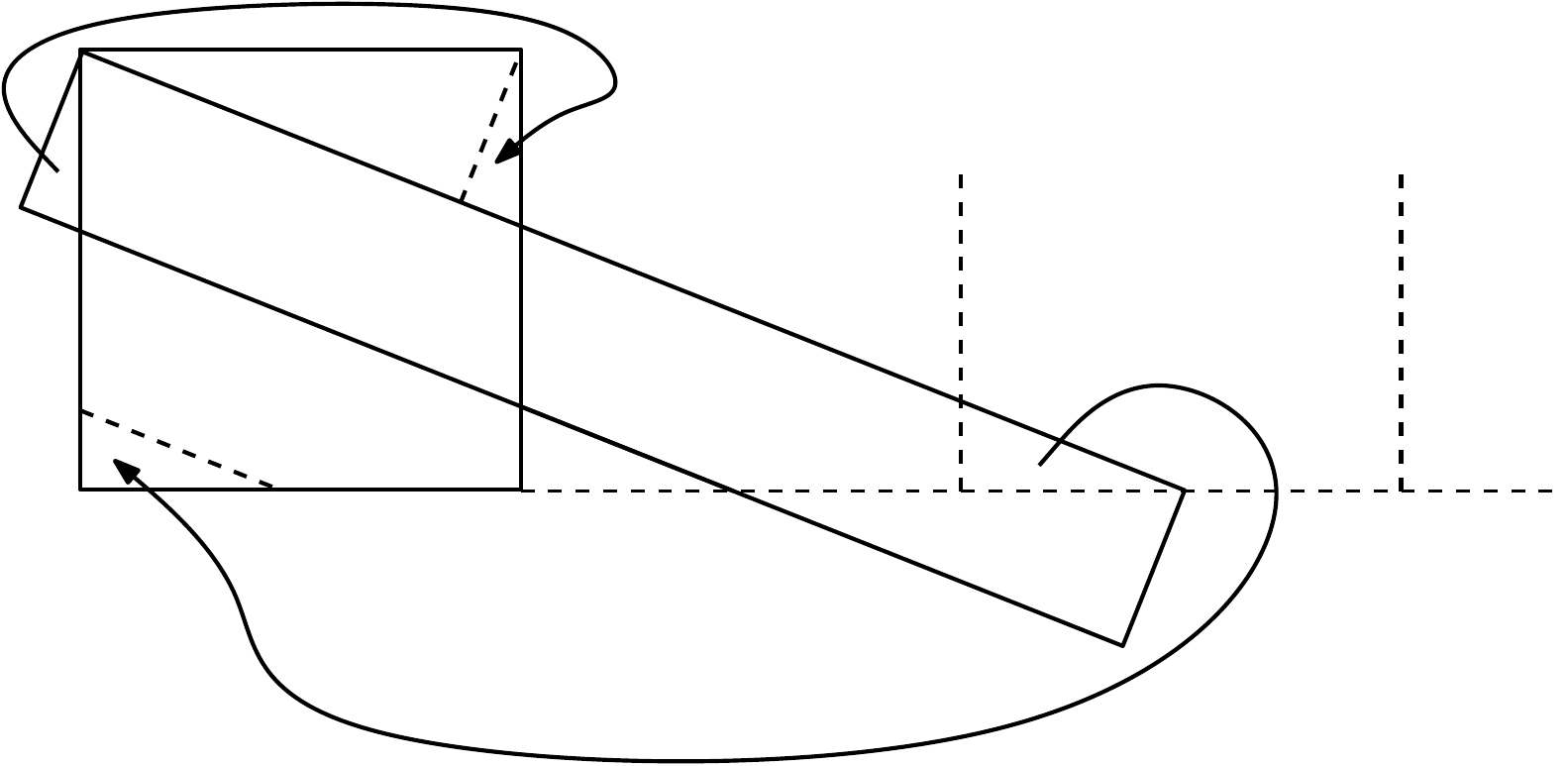}
\caption{Montucla's Dissection}
\label{fig:montucla1}
\end{figure}

We are interested in generalizing this dissection to dimensions greater than two, i.e., in dissections from the brick $[0,a_1] \times \ldots \times [0,a_n]$, $a_1a_2 \ldots a_n = 1$, into the cube $[0,1]^n$, for $n>2$. 
In this paper, we show that Montucla's dissection is a manifestation of a general dissection principle called the \textit{Two Tile Theorem} \cite{sloane2009generalizations}. The theorem states that if two polytopes tile the same space, under the action of the same isometry group, then they must be equidissectable. We generalize Montucla's dissection by finding an explicit translation group (point lattice) $\Lambda$ such that, under the action of $\Lambda$, the brick and the cube tile $\mathbb{R}^n$. 
We then design an algorithm that computes the dissection with \textit{linear} complexity $O(n)$. This follows from an explicit construction of a bidiagonal generator matrix for the group of translations used in our application of the Two Tile Theorem.

We note that the process described in \cite[p.222]{frederickson2003dissections} is for a  rectangle-to-rectangle dissection. For simplicity (and for applications) we will only consider cube-to-brick generalizations of this dissection. A brick-to-brick dissection can be obtained with linear complexity $O(n)$ by composing a cube-to-brick dissection with a brick-to-cube dissection. 

\section{Motivation and Related Works}
While dissections arise in geometry and in recreational mathematics, the connections with coding have remained largely unexplored. In~\cite{tian2009coding}, the problem of encoding binary sequences of constant Hamming weight is shown to be related to that of constructing a brick to tetrahedron dissection, and an efficient algorithm for accomplishing this is presented. The cuts in the resulting dissection are  perpendicular to the coordinate axes. The tetrahedron in~\cite{tian2009coding} belongs to the family of Hill tetrahedra. A generalization of Schobi's dissection~\cite{Schobi:1985} of a three-dimensional Hill tetrahedron  to dimensions greater than three is presented in \cite{sloane2009generalizations}. The generalization is based on an application of the Two Tile Theorem as stated in that paper.  A  tetrahedron-to-brick dissection based on a non-rigid transformation and a solution to the problem of constructing  $\mathbb{Z}^n$ to $\mathbb{Z}^n$ mappings that arise when cuts are not perpendicular to the coordinate axes is presented in~\cite{vs:2006}. 
 
In \cite{campello2013projections}, a coding scheme for transmitting a real vector drawn uniformly from the cube $[0,1]^k$ over an $n$-dimensional Gaussian channel, $k < n$, is proposed. It is shown that this can be accomplished by dissecting the cube into a suitably chosen brick, and then encoding the points in the brick. To do so, given a point in the cube, one has to efficiently compute the map that takes it into the corresponding piece in the brick. This was the main motivation for our algorithm for computing the cube-to-brick dissection, albeit the problem has a mathematical interest on its own. 

Constructions of rectangle-to-square and brick-to-cube (in $\mathbb{R}^3$) dissections can be found in \cite{frederickson2003dissections}. In a different context, techniques of dissecting a rectangle into a finite number of non-overlapping squares can be found in \cite{classicPapers}, whereas conditions for the triangularization of an $n$-dimensional cube into several simplexes are derived in \cite{simplexes}.

\section{Preliminaries}
\label{sec:preliminaries}

Let $P \subset \mathbb{R}^n$ be a polytope. A \textit{dissection} of $P$ into $m$ different polytopes is a decomposition of the form 
$$P = \bigcup_{i=1}^m P_i,$$
such that the interiors of $P_i$ and $P_j$ are disjoint for any $i \neq j$ and $m$ is finite. The elements $P_i$ are called \textit{pieces} of the dissection. Let $Q \subset \mathbb{R}^n$ be another polytope. Suppose that there is a dissection
$$Q = \bigcup_{i=1}^m Q_i,$$
such that $Q_i$ may be obtained from $P_i$ through an isometry $\phi_i$ of $\mathbb{R}^n$. We say that $Q$ and $P$ are \textit{equidissectable}. Equidissectable polytopes have the same volume, but for dimensions greater than two, equality of volume is not a sufficient condition for equidissectability. This subject, of fundamental interest to geometers, is treated in detail in \cite{Bol:1978}. 
We note that our interest lies in  the bijection implicit in the dissection, i.e., given a point $\xx \in {P}_i$, find $\yy = \phi_i(\xx) \in Q_i$. 

The Two Tile Theorem, which underlies our construction, is described next.

Let $G$ be an isometry group in $\mathbb{R}^n$, $P$ a polytope, and $\Omega \subseteq \mathbb{R}^n$. If the images of $P$ by the action of $G$ have disjoint interiors and $\Omega = \bigcup_{g \in G} gP$, we say that $P$ is a $G$-tile for $\Omega$.
\begin{theorem}[Two Tile Theorem \cite{sloane2009generalizations}] If for some set $\Omega \subseteq \mathbb{R}^n$ and some isometry group $G$ of $\mathbb{R}^n$, two polytopes $P$ and $Q$ are $G$-tiles for $\Omega$, then $P$ and $Q$ are equidissectable.
\label{thm:twotiles}
\end{theorem}

Important isometry groups for our purpose are translation groups obtained from lattices. A (full-rank) \textit{lattice} $\Lambda$ is a discrete additive subgroup of $\mathbb{R}^n$ not contained in any proper subspace. Any lattice has a generator matrix, i.e., a full-rank matrix $B$ such that ${\Lambda = \left\{ \bm{u} B  : \bm{u} \in \mathbb{Z}^n \right\}}$. The vectors of $\Lambda$ may be regarded as translations, and thus $\Lambda$ acts on a polytope by translating it by its vectors, so that
$$\Omega = \bigcup_{\xx \in \Lambda} \xx P = \bigcup_{\xx \in \Lambda} (\xx + P).$$
Two $\Lambda$-tiles for $\mathbb{R}^n$ are equidissectable by Thm. \ref{thm:twotiles}.

\begin{remark}
Some simple examples of $\Lambda$-tiles are the \textit{fundamental parallelotope} $\mathcal{P}$ and the \textit{Voronoi region} $\mathcal{V}$:
$${\mathcal{P} := \left\{ \bm{u} B : 0 \leq u_i \leq 1, i = 1, \ldots n \right\}} \mbox{ and }$$ $$ \mathcal{V} := \left\{ \xx \in \mathbb{R}^n: \left\| \bm{x} \right\| \leq \left\| \bm{x} - \bm{y} \right\|, \mbox{ for all } \yy \in \Lambda \right\}.$$
Thm. \ref{thm:twotiles} then implies that a Voronoi region is equidissectable with the parallelotope $\mathcal{P}$. Since a parallelotope is equidissectable with a brick, this implies that the Voronoi region is also equidissectable with a brick. From our brick-to-cube dissection, this implies that a Voronoi region is always equidissectable with a cube. 

\end{remark}
We call the cartesian product set $\mathcal{R} = [0,a_1] \times [0,a_2] \times \ldots \times [0,a_n]$ a \textit{brick} with lengths $a_1,\ldots,a_n$. The image of $\mathcal{R}$ through an isometry is called a \textit{realization} of the brick. In general, the image of a polytope $P$ by an isometry of $\mathbb{R}^n$ is called a \textit{realization} of $P$. 

We say that an algorithm \textit{computes} the dissection if, for some realizations $\tilde{P}$ and $\tilde{Q}$ of $P$ and $Q$, respectively, and given a point $\xx \in \tilde{P}$, it outputs the point $\phi(x) \in \tilde{Q}$, under the bijection $\phi: \tilde{P} \to \tilde{Q} $ implicit in the dissection. Note that by this definition, we are free to choose the specific realizations--in fact, some realizations may be simpler than others (see Remark \ref{rmk:vanilla} for the specific case of our cube-to-brick dissection).

\section{Montucla's Dissection Revisited}
We give a formal description of the dissection in terms of the Two Tile Theorem. For simplicity, we will consider that the square and the rectangle have unit area.

Let $\Lambda = \{ (u_1 + u_2 \beta, u_2): u_1, u_2 \in \mathbb{Z}\}$ be the lattice generated by matrix

\begin{equation}
B =  \left( \begin{array}{cc} 1 & 0 \\ \beta & 1 \end{array} \right).
\label{eq:translationGroupR2}
\end{equation}
Now consider the orthogonal pair of vectors:

$$\bb_1 = \left(\frac{1}{1+\beta^2},\frac{-\beta}{1+\beta^2}\right) \mbox{ and }\bb_2 = (\beta,1).$$
Let 
$$\mathcal{R} = \left\{ \alpha_1 \bb_1 + \alpha_2 \bb_2 : 0 \leq \alpha_1 \leq 1, 0 \leq \alpha_2 \leq 1 \right\}.$$
be a realization of the rectangle with lengths $\sqrt{\beta^2+1}$ and $1/\sqrt{\beta^2+1}$. The following proposition shows a dissection between $\mathcal{R}$ and the square $\mathcal{C} = [0,1]\times[0,1]$.

\begin{prop} The square $\mathcal{C}$ and the rectangle $\mathcal{R}$ are $\Lambda$-tiles for $\mathbb{R}^2$.
\label{prop:mont}
\end{prop}
\begin{proof}
(i) $\mathcal{C}$ is a $\Lambda$-tile for $\mathbb{R}^2$.

We first prove that the translations of $\mathcal{C}$ by vectors of $\Lambda$ are interior disjoint. For let $\xx_1, \xx_2 \in \Lambda$. Suppose there is an element $\bm{v}$ in the interior of $\xx_1 + \mathcal{C}$ and $\xx_2 + \mathcal{C}$. We must have $\bm{v} = \xx_1 + \cc_1 = \xx_2 + \cc_2$ for some $\cc_1, \cc_2 \in \mbox{int}{(\mathcal{C})}$. This implies $\xx_1 - \xx_2 = \cc_1 - \cc_2 \in (-1,1) \times (-1,1)$, but since $\xx_1 - \xx_2$ are in $\Lambda$ we can write it as $\xx_1 - \xx_2 = (u_1 + u_2 x, u_2)$ for $u_1, u_2 \in \mathbb{Z}$. From this, $u_2 \in \mathbb{Z} \cap (-1,1)$, which implies $u_2 = 0$ and consequently $u_1 = 0$, i.e., $\xx_1 = \xx_2$. Hence the interiors of $\xx_1 + \mathcal{C}$ and $\xx_2 + \mathcal{C}$ are disjoint for $\xx_1 \neq \xx_2$.

Now, we must prove that $\Lambda + \mathcal{C} = \mathbb{R}^2$, i.e., that any point $\yy \in \mathbb{R}^2$ may be written as $\yy = \uu B + \bm{c}$ where $\uu \in \mathbb{Z}^2$. This is true by setting $u_2 = \left\lfloor x_2 \right\rfloor$, $c_2 = x_2 - u_2$, $u_1 = \left\lfloor x_1 - u_2 \beta\right \rfloor$ and $c_2 = x_1 - \left\lfloor x_1 - u_2 \beta \right\rfloor - u_2 \beta$.

(ii) $\mathcal{R}$ is a $\Lambda$-tile for $\mathbb{R}^2$. Let $\xx_1$ and $\xx_2$ be vectors of $\Lambda$ and suppose there is a point inside the interior of both translated rectangles $\xx_1 + \mathcal{R}$ and $\xx_2 + \mathcal{R}$. Then there are integers $u_1, u_2$ and reals $\alpha_1, \alpha_2 \in (-1,1)$ such that $\xx_1 - \xx_2 = (u_1 + u_2 x, u_2) = \alpha_1 \bb_1 + \alpha_2 \bb_2$, or:
\begin{eqnarray*}
\lefteqn{u_1(1,0)+ u_2(\beta, 1) = }  \\
&  & \alpha_1\left(\frac{1}{1+\beta^2},\frac{-\beta }{1+\beta^2}\right) + \alpha_2(\beta ,1) = \alpha_1(1,0) + \left(\alpha_2 - \frac{\alpha_1 \beta}{\beta^2+1}\right)(\beta,1).
\end{eqnarray*}
From the equation above, $u_1 = \alpha_1 \in (-1,1)$ which implies $u_1 = 0$ and subsequently $u_2 = 0$. 
Hence the interiors of translations of $\mathcal{R}$ by $\Lambda$ are disjoint. Again, one can check that any point in $\bm{y} \in \mathbb{R}^2$ can be written as $\bm{y} = \bm{r} + \bm{x}$ where $\bm{x} \in \Lambda$ and $\bm{r} \in \mathcal{R}$, proving statement (ii). 
\qed
\end{proof}

From the Two Tile Theorem, and Proposition \ref{prop:mont}, $\mathcal{R}$ and $\mathcal{C}$ are equidissectable. Up to a change of orientation, this is precisely Montucla's Dissection (see Fig. \ref{fig:tessel}).

To dissect any rectangle of sides $a, 1/a$, we choose a realization $\mathcal{R}$ of it with $\beta = \pm \sqrt{a^2 - 1}$. The pieces can be described as the intersection of the rectangle $\mathcal{R}$ with translations of the square by $\Lambda$ (or vice-versa). This way, we can count the number of pieces as the number of translation vectors $\xx \in \Lambda$ such that $([0,1]^2+\xx) \cap \mathcal{R} \neq \emptyset$ to show that only $2+\left\lceil \sqrt{a^2-1}\,\ \right\rceil \leq \left\lceil a \right\rceil + 2$ pieces are needed.

\begin{figure}
\centering
\subfloat{\includegraphics[scale=0.43]{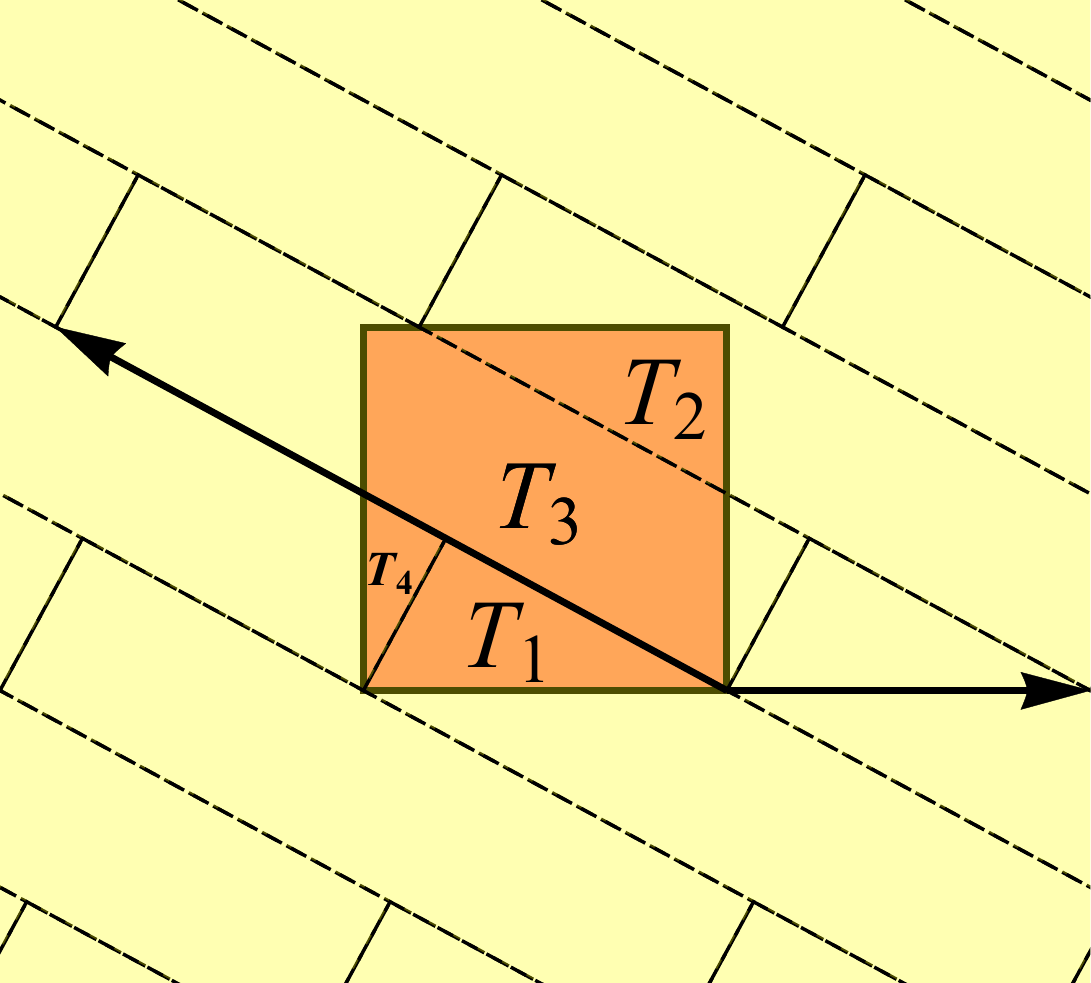}}\hspace{1cm}
\subfloat{\includegraphics[scale=0.35]{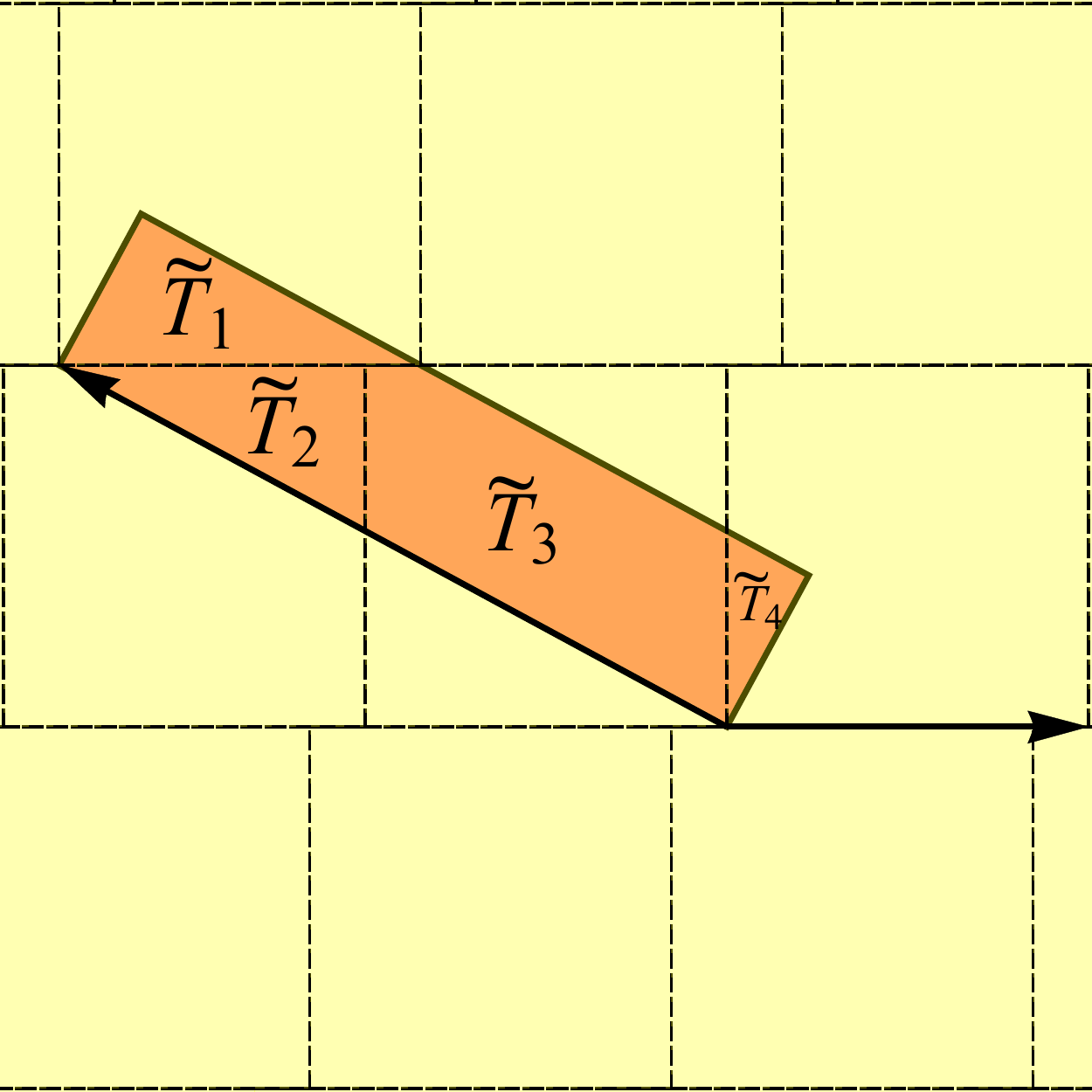}}
\caption{Montucla's Dissection in terms of the Two Tile Theorem}
\label{fig:tessel}
\end{figure}

\section{Higher Dimensions}
\begin{lemma}
Let $\Lambda \subset \mathbb{R}^n$ be the lattice with lower triangular generator matrix $B$,
\begin{equation} B = \left( \begin{array}{ccccc} 1 &  0 &  \ldots & 0 \\ b_{21} & 1 &  \ldots & 0 \\ \vdots  & \ddots & \ddots & \vdots \\
 b_{n1} & \ldots & b_{n,n-1} & 1\end{array}\right).
\label{eq:Triangular}
\end{equation}
The cube $\mathcal{C} = [0,1]^n$ is a $\Lambda$-tile for $\mathbb{R}^n$.
\label{lem:1}
\end{lemma} 
\begin{proof}Let $\xx_1, \xx_2 \in \Lambda$. If there is an element in the interior of $(\bm{x}_1 + \mathcal{C}) \cap (\bm{x}_2 + \mathcal{C})$, then $\bm{x}_1 - \bm{x}_2 \in (-1,1)^n$. But $\bm{x}_1 - \bm{x}_2$ belongs to $\Lambda$, therefore there is an integer vector $\bm{u}$ such that $\bm{u} B = \bm{x}_1 - \bm{x}_2 \in (-1,1)^n$. Since $B$ is lower triangular with unitary diagonal, we have ${u}_n \in \mathbb{Z} \cap (-1,1)$ which gives ${u}_n = 0$. From this, we have ${u}_n b_{n, n-1} + {u}_{n-1} = \bm{u}_{n-1} \in \mathbb{Z} \cap (-1,1)$ which implies that ${u}_{n-1} = 0$. Continuing this process we prove that $\bm{x}_1 = \bm{x}_2$, and hence translations of $\mathcal{C}$ by $\Lambda$ are non-overlapping. Furthermore, any point $\bm{y} \in \mathbb{R}^n$ can be written as $\bm{y} = \bm{x} + \bm{u}B$ where $\bm{x} \in [0,1]^n$, and this completes the proof. \qed
\end{proof}
It is worth noting that the converse of the above lemma is also true. If $\mathcal{C}$ is a $\Lambda$-tile for $\mathbb{R}^n$, then, up to a change of coordinates, $\Lambda$ has a generator matrix in form \eqref{eq:Triangular}. This follows as a consequence of Minkowski's Conjecture that in any lattice tiling of $\mathbb{R}^n$ by cubes, there exist two cubes that share a complete face (see e.g. \cite{TwinCubes}). The conjecture was proved by H\'ajos in 1941 \cite{Hajos}.

For our main theorem, we need a slight generalization of Lemma \ref{lem:1}, that can be proved in a similar way.

\begin{lemma} Let $U$ be any upper triangular matrix with $u_{ii} = 1$ for all $i = 1, \ldots, n$. Then the parallelogram $\mathcal{R} = \left\{\bm{\alpha} U B : \alpha \in [0,1]^n\right\}$ is a $\Lambda$-tile for $\mathbb{R}^n$.
\label{lem:2}
\end{lemma}

A consequence of this second lemma is that the brick produced by the Gram-Schmidt orthogonalization of the rows of $B$ starting from the last row is also a $\Lambda$-tile. To see this, consider a matrix $R$ whose rows $\bm{r}_1, \ldots, \bm{r}_n$ are the orthogonalized vectors:
$$\bm{r}_n = \bm{b}_n$$
$$\bm{r}_i = \bm{b}_i - \sum_{j=i+1}^n \frac{\left\langle \bm{b}_i, \bm{r}_j \right\rangle}{\left\langle\bm{r}_j, \bm{r}_j \right\rangle} \bm{r}_j, \mbox{ for } i = n-1,\ldots, 1.$$
Let $A$ be the coefficient matrix of the orthogonalization, i.e., 
\begin{equation} a_{ij} = \displaystyle \frac{\left\langle \bm{b}_i, \bm{r}_j \right\rangle}{\left\langle\bm{r}_j, \bm{r}_j \right\rangle} \mbox{ for } j > i
\label{eq:GramSchmidt}
\end{equation}
and $a_{ii} = 1$. 
We have $B = AR$. Taking $U = A^{-1}$ in Lemma \ref{lem:2} we obtain that the brick $\mathcal{R} = \left\{ \alpha_1 \bm{r}_1 + \ldots \alpha_n \bm{r}_n : 0 \leq \alpha_i \leq 1 \right\}$ is a $\Lambda$-tile for $\mathbb{R}^n$.

We are now in position to prove the equidissectability of a brick and a cube. We first provide a simple proof, by inductively applying Montucla's dissection, and then provide a more constructive proof, that will yield our linear-time  algorithm for computing the dissection.

\begin{theorem} The brick with lengths $a_1, a_2, \ldots, a_n$ such that $a_1 a_2 \ldots a_n=1$ and the cube $\mathcal{C} = [0,1]^n$ are equidissectable.
\label{thm:bricktocube}
\end{theorem}
\textit{Proof 1.} By induction on the dimension. For $n = 2$, from the last section, it is true that the unit area rectangle and the square are equidissectable (this means that any two unit area rectangles are also equidissectable). If the theorem is valid for $n$, then any two $n$-dimensional bricks with same area are equidissectable. Let $a_1, a_2, \ldots, a_n, a_{n+1}$ be the lengths of a brick. By the induction hypothesis, we can dissect the brick $[0, a_1] \times [0, a_2] \times \ldots \times [0,a_n] \times [0,a_{n+1}]$ into the brick $[0,1] \times [0,1] \times .. [0,1] \times \ldots [0, a_1 a_2 \ldots a_n] \times [0, a_{n+1}]$ by cutting along the first coordinates. Now, using Montucla's dissection, cut along the last two coordinates and rearrange it into the unit cube.
\\ \qed
\noindent \textit{Proof 2.} We show an explicit dissection with one step of the Two Tile Theorem. Suppose without loss of generality that $a_1 \leq a_2 \leq ... \leq a_n$, and so $\prod_{k=j} a_k \geq 1$ for all $1\leq j \leq n$. Let $B$ be the matrix given by
\begin{equation}b_{ij} = \left\{ \begin{array}{cc} 1 & \mbox{ if } i = j \\ \displaystyle \frac{\sqrt{ \prod_{k=j}^n a_k^2 -1}}{\prod_{k=j+1}^n a_k} & \mbox{ if }i = j+1 \\ 0 & \mbox{ otherwise}
\end{array}\right.
\label{eq:tridiagonalMatrix}
\end{equation}
For example, for $n = 4$:

\begin{equation}B = \left(
\begin{array}{cccc}
 1 & 0 & 0 & 0 \\
 \frac{\sqrt{a_2^2 a_3^2 a_4^2-1}}{a_3 a_4} & 1 & 0 & 0 \\
 0 & \frac{\sqrt{a_3^2 a_4^2-1}}{a_4} & 1 & 0 \\
 0 & 0 & \sqrt{a_4^2-1} & 1 \\
\end{array}
\right)
\label{eq:GoodMatrix}
\end{equation}
Applying Gram-Schmidt orthogonalization on the rows of $B$ starting from the last row, we obtain a matrix $R$ whose rows are orthogonal vectors. In what follows, we prove that $\left\| \bm{r}_j \right\| = a_j$ for all $j = 1, \ldots, n$. In fact, for $j = n$ this is true since $\bm{r}_n = \bm{b}_n$. Suppose it is true for $n,n-1, \ldots j+1$.  Then, due to the structure of matrix $B$, we have $\bm{b}_j \bm{r}_i = 0$ for $i > j+1$, and 

$$\left\|\bm{r}_{j}\right\|^2 = \left\|\bm{b}_j -\frac{\left\langle\bm{b}_j,\bm{r}_{j+1}\right\rangle}{\left\langle \bm{r}_{j+1},\bm{r}_{j+1}\right\rangle} \bm{r}_{j+1}\right\|^2 = \frac{ \prod_{k=j}^n a_k^2 -1}{\prod_{k=j+1}^n a_k^2} + \left\|\bm{e}_{j+1} - \frac{\left\langle \bm{e}_{j+1},\bm{r}_{j+1} \right\rangle}{\left\langle \bm{r}_{j+1} \bm{r}_{j+1}\right\rangle} \bm{r}_{j+1} \right\|^2$$
where $\bm{e}_j$ is the canonical vector with one in the $j$-th position. After some simplifications:

$$\left\|\bm{e}_{j+1} - \frac{\left\langle \bm{e}_{j+1},\bm{r}_{j+1} \right\rangle}{\left\langle \bm{r}_{j+1} \bm{r}_{j+1}\right\rangle} \bm{r}_{j+1} \right\| = \frac{1}{a_{j+1} a_{j+2} \ldots a_n} \mbox{, and therefore } \left\| \bm{r}_j \right\| = a_j.$$
 This proves that $\mathcal{R} = \left\{ \alpha_1 \bm{r}_1 + \ldots \alpha_n \bm{r}_n : 0 \leq \alpha_i \leq 1 \right\}$ is a realization of the brick with lenghts ${a_1}, \cdots, a_n$. Then, from Lemma 1, the cube $[0,1]^n$ is a $\Lambda$-tile, from Lemma 2, the rectangle $\mathcal{R}$ is a $\Lambda$-tile and from what have been proven $\mathcal{R}$ is a realization of the brick. By invoking the Two Tile theorem, the proof is concluded.
\\ \qed

Let $A$ be the matrix that captures the Gram-Schmidt coefficient of the orthogonalization of $B$ (Eq. \eqref{eq:tridiagonalMatrix}), as in \eqref{eq:GramSchmidt}. The central observation for our linear-time algorithm to compute the cube-to-brick dissection is the fact that both $B$ and $A$ are bidiagonal (i.e., all elements are zero except the main diagonal and either the diagonal above or the diagonal below). In fact, a quick induction shows a closed form expression for $A$:

\begin{equation}a_{ij} = \left\{ \begin{array}{cc} 1 & \mbox{ if } i = j \\ \displaystyle \frac{\sqrt{ \prod_{k=j}^n a_k^2 -1}}{a_{j}^2 \prod_{k=j+1}^n a_k } & \mbox{ if }j = i+1 \\ 0 & \mbox{ otherwise.}\end{array} \right.
\label{eq:GramSchmidtMatrix}
\end{equation}

\begin{algorithm}{Compute the Dissection}
\caption{Compute the Dissection}
\begin{algorithmic}[1]
\Procedure{Cube-to-Brick}{$\bm{x} \in \mathcal{C}, a_1,\ldots,a_n$}
\State Find $\bm{z}$ such that $\bm{z} B = \bm{x}$
\State $\bm{\overline{x}} \gets \bm{z}A$
\State $u_1 \gets \left\lfloor \overline{x}_1 \right\rfloor$
\For {$i:2,\ldots,n$}
\State $u_i = \left\lfloor \overline{x}_j - u_j a_{j,j+1}  \right\rfloor$ 
\EndFor 
\State \textbf{return} $\bm{y} = \bm{x} - \bm{u} B$.
\EndProcedure
\end{algorithmic}
\label{alg:algorithm}
\end{algorithm}

\begin{theorem}
There is an algorithm that computes the dissection with $O(n)$ arithmetic operations.
\label{thm:linear}
\end{theorem}
\begin{proof}We prove that Algorithm \ref{alg:algorithm} is correct and can be performed with $O(n)$ operations. Let $A$ and $B$ be as in \ref{eq:tridiagonalMatrix} and \ref{eq:GramSchmidtMatrix}. Let $R$ be as in Proof 2 of Thm. \ref{thm:bricktocube}. \\[1\baselineskip]
\noindent \textit{Correctness of the algorithm:} To find the find the right piece of the dissection, one has to find the translation $\bm{w} \in \Lambda$ such that $\bm{x} = \bm{y} + \bm{w}$ with $\bm{y} \in \mathcal{R}$. We show that $\bm{w}=\bm{u}B$ is this translation, and therefore $\bm{x} - \bm{u} B$ belongs to $\mathcal{R}$, and corresponds to the image of $\bm{x}$ in the bijection determined by the dissection.

From steps $2$ and $3$ of the algorithm, $\bm{\overline{x}} = \bm{x} B^{-1} A$ and hence:
$$\bm{y}R^{-1} = \bm{x} R^{-1} - \bm{u} B R^{-1} = \bm{\overline{x}} - \bm{u} A.$$
From the definition of $\bm{u}$, we have that $\bm{\overline{x}} - \bm{u} A \in [0,1]^n$. Hence, $\bm{y} = \bm{\alpha} R$ for some $\bm{\alpha} \in [0,1]^n$, proving the statement.\\
\indent \textit{Complexity:} First observe that $A$ and $B$ are bidiagonal matrices with non-zero elements only on the main and second upper/lower diagonal. Then the linear system of equations in step $2$ and the vector-matrix multiplications in steps $3$ and $8$ can be computed with $O(n)$ operations. For steps $4$ and $5$ we need $n$ round operations and $n-1$ multiplications/additions. This gives a total of $O(n)$ arithmetic operations.
\end{proof}
\qed


\begin{remark} The output of Algorithm \ref{alg:algorithm} is a point $\yy$ in the realization $\mathcal{R}$ of the brick described in the proof of Thm. \ref{thm:bricktocube}. Recovering the coordinates $\alpha$ of the point in the ``canonical'' system (i.e., in the ``right'' brick $[0,a_1] \times \ldots [0,a_n]$) can be as well performed in linear time by exploiting the decomposition of matrix $R = A^{-1} B$. To do so, we must find $\alpha$ such that $\yy = \alpha R = \alpha A^{-1} B$. We first find $\ww$ such that $\yy = \ww B$ and then perform matrix multiplication $\ww A = \alpha$. Both the multiplication and solving the linear system can be performed in linear time, since $A$ and $B$ are bidiagonal.
\label{rmk:vanilla}
\end{remark}

\section{Application: Analog Mappings}
In \cite{campello2013projections}, maps between a $k$-dimensional uniform source on $[0,1]^k$, to be transmitted over an $n$-dimensional Gaussian channel are studied. The objective is to minimize the mean squared error (MSE) criterion. 

While the source support is originally a cube, the techniques in \cite[Sec. VI]{campello2013projections} show that it is possible to enhance the MSE significantly when transmitting points of a modified support (e.g., a brick). It is also shown that dissections of polytopes can be used to map the original support into the modified one without degrading the MSE. To this purpose, one needs to, given a point in the dissected cube, find the corresponding point in the brick. Thm. \ref{thm:linear} shows that it is possible to do this with $O(n)$ operations.

The dimensions of the brick in \cite{campello2013projections} are related to the signal-to-noise (SNR) ratio of channel. Other important parameters in the coding scheme \cite{campello2013projections} are a ``shrinking factor $(1-\varepsilon)$'' and translation vectors $\bm{t}_i$, used to prevent anomalous errors, which are not addressed here.

\section{Conclusion}
This paper shows a dissection from a cube into a brick in $\mathbb{R}^n$. The dissection is a generalization of a two-dimensional method by Montucla. As a by-product of the generalization, we show that computing the dissection can be performed with as few as $O(n)$ operations.

There are some rectangle-to-square dissections that require fewer pieces than Montucla's (see \cite[p.222, 237-243]{{frederickson2003dissections}}). We do not know whether these are also a manifestation of the Two Tile Theorem, neither if it is possible to generalize them to higher dimensions, which is left as an open question. 



\bibliography{curves}
\bibliographystyle{alpha}
\end{document}